\documentclass[11pt,reqno]{amsart}

\usepackage[margin=1.15in]{geometry}
\usepackage{amsmath,amssymb,amsthm}
\usepackage{mathtools}
\usepackage{booktabs}
 \usepackage{tikz}

\usepackage{xcolor}
\definecolor{darkred}{rgb}{0.45,0,0}
\definecolor{darkblue}{rgb}{0.2, 0.2, 0.9}
\definecolor{darkorange}{rgb}{0.8, 0.4, 0.1}
\definecolor{darkgreen}{rgb}{0.2, 0.7, 0.2}
\definecolor{brickred}{rgb}{0.9,0,0}
\definecolor{pink}{rgb}{0.95, 0.6, 0.7}
\definecolor{palepink}{rgb}{0.98, 0.85, 0.88}
\definecolor{palepurple}{rgb}{0.92, 0.85, 0.95}

\definecolor{shadecolor}{rgb}{1,0.8,0.3}
\definecolor{myurlcolor}{rgb}{0.5,0,0}
\definecolor{mycitecolor}{rgb}{0,0,0.9}
\definecolor{myrefcolor}{rgb}{0,0,0.8}
\definecolor{hyperrefcolor}{rgb}{0.5,0,0}

\usepackage[
    colorlinks, 
    citecolor=blue, 
    urlcolor=darkred, 
    final, 
    hyperindex, 
    pagebackref, 
    linkcolor = darkblue
]{hyperref}

\newcommand{\define}[1]{{\bf \boldmath{#1}}\index{#1}}

\newtheorem{theorem}{Theorem}
\newtheorem{proposition}[theorem]{Proposition}
\newtheorem{lemma}[theorem]{Lemma}

\theoremstyle{definition}

\theoremstyle{remark}

\newcommand{\C}{\mathbb{C}}
\newcommand{\R}{\mathbb{R}}
\newcommand{\Z}{\mathbb{Z}}

\newcommand{\g}{\mathfrak{g}}
\newcommand{\h}{\mathfrak{h}}
\newcommand{\e}{\mathfrak{e}}
\newcommand{\so}{\mathfrak{so}}
\renewcommand{\sl}{\mathfrak{sl}}

\newcommand{\gsm}{\g_{\mathrm{SM}}}
\newcommand{\sm}{\mathrm{SM}}
\newcommand{\m}{\mathfrak{m}}
\newcommand{\s}{\mathfrak{s}}

\newcommand{\SU}{\mathrm{SU}}
\newcommand{\U}{\mathrm{U}}
\newcommand{\A}{\mathrm{A}}
\newcommand{\D}{\mathrm{D}}
\newcommand{\E}{\mathrm{E}}
\newcommand{\GSM}{G_{\mathrm{SM}}}

\newcommand{\ip}[2]{\langle #1, #2 \rangle}
\newcommand{\gen}{\mathrm{gen}}
\newcommand{\ad}{\mathrm{ad}}
\newcommand{\maps}{\colon}

\newcommand{\ubar}{\overline{u}} 
\newcommand{\dbar}{\overline{d}} 
\newcommand{\nubar}{\overline{\nu}} 
\newcommand{\cbar}{{\overline{c}}} 

\title{Three Generations In E$_7$}
\author{John C.\ Baez}
\address{
School of Mathematics, University of Edinburgh, \\
James Clerk Maxwell Building, Peter Guthrie Tait Road, Edinburgh, UK EH9 3FD}
\address{
Department of Mathematics, University of California, Riverside CA, USA 92521 }
\date{\today}

\begin{document}

\maketitle
\vspace{-3em}

\begin{abstract}
\noindent
Starting from the Standard Model Lie algebra $\gsm = \sl(3) \oplus \sl(2) \oplus \C$
sitting inside the complex Lie algebra $\e_7$, we show how to decompose $\e_7$ 
into the direct sum of a Lie subalgebra containing $\gsm$ and three 32-dimensional 
subspaces, each of which forms the same representation of $\gsm$ as one generation of 
fermions and their antiparticles.   The setting is due to 
Nasmith, and much of the mathematics is that underlying the $\E_7$ generation 
unification of Kugo and Yanagida.  New features include the derivation, in which 
as many results as possible rely only on the embedding $\gsm \subset \e_7$, and 
also the description of each 32-dimensional subspace as a copy of the exterior algebra 
$\Lambda\C^5$. 
\end{abstract}

\section{Introduction}

The complex exceptional Lie algebra $\e_7$ is large enough to contain the complexified Standard
Model Lie algebra
\[
  \gsm \;=\; \C \oplus \sl_2 \oplus \sl_3
\]
with a great deal of room to spare, since
$\dim\e_7 = 133$.   The question is whether this room can be filled in an interesting
way.   Nasmith \cite{Nasmith2020,Nasmith2023} argued that it 
can, showing that $\gsm$ acts, via the Lie bracket of $\e_7$, on three 32-dimensional 
subspaces of $\e_7$ in a manner that precisely matches its representation on three 
generations of Standard Model fermions and their antiparticles--- including right-handed 
neutrinos and their antiparticles.

What follows is a somewhat new account of this story.   Starting from a suitable 
embedding of $\gsm$ in $\e_7$, we construct a Lie subalgebra 
\[   \gsm \subset \sl_6 \oplus \C^2 \subset \e_7 \]
and a vector space isomorphism
\[            \e_7 \; \cong \; (\sl_6 \oplus \C^2) \oplus V \]
where the space $V$ has dimension $3 \times 32$.  The subalgebra 
$\gsm$ acts on $V$, via the $\e_7$ Lie bracket, 
precisely as it does on three generations of Standard Model fermions and their antiparticles.   

Like Nasmith, we are not proposing a theory of physics.
We are only observing a fascinating mathematical pattern that might (or might not)
be of some use in physics.    

We tell our story in the language of complex
Lie algebras rather than the language of compact Lie groups because it is more
convenient for our calculations.   However, everything we say can be translated into a story about the compact real form of $\E_7$ and the Standard Model gauge group
\begin{equation}
\label{eq:GSM}
\begin{array}{ccl} 
\GSM &= &
\Big\{ x \in \SU(5) : x = 
\left( 
\begin{array}{c c c c c}
\ast & \ast & 0 & 0 & 0 \\
\ast & \ast & 0 & 0 & 0 \\
0 & 0 & \ast & \ast & \ast \\
0 & 0 & \ast & \ast & \ast \\
0 & 0 & \ast & \ast & \ast 
\end{array}
\right) \; \Big\}  \\ \\
&\cong &
\left(\U(1) \times \SU(2) \times \SU(3)\right) / \Z_6 .
\end{array}
\end{equation}

Note that this equation presents $\GSM$ as a subgroup of $\SU(5)$.   It is well known \cite{BaezHuerta2010} that if we restrict the natural representation of $\SU(5)$ on
the exterior algebra $\Lambda \C^5$ to this subgroup, we get exactly the usual 32-dimensional representation 
of $\GSM$ on one generation of Standard Model fermions and their antiparticles.   This
in turn gives a representation of $\gsm \subset \sl_5$ on $\Lambda \C^5$.  We call this the
\define{Standard Model representation} of $\gsm$.

While elegant, this construction does not address the \emph{three} generations of the Standard 
Model.    This is where $\e_7$ comes in.   It is worth sketching how this works before giving full details in the body of the paper.

We first fix a suitable inclusion of $\gsm$ in $\e_7$.   There is then a unique $\sl_3$ subalgebra of $\e_7$ that commutes with everything in $\gsm$.   This will act as symmetries relating the three generations, so we call it $\sl_3^\gen$.  

We next choose a Cartan subalgebra of $\e_7$ compatible with that of $\gsm$ and $\sl_3^\gen$.  This gives a way to split the eight dimensions of $\sl_3^\gen$ into two Cartan directions and six root spaces.    In the Cartan of $\sl_3^\gen$, the six roots lie on three lines:

\vskip 1em
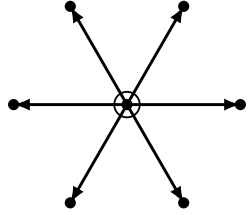
\begin{center}
\begin{tikzpicture}[scale=1.5,>=latex]

  \def\R{1}

  \foreach \a in {0,60,120,180,240,300}{
    \draw[->,line width=1.1pt,black] (0,0) -- (\a:\R);
    \fill (\a:\R) circle (1.4pt);
  }

  \fill (0,0) circle (1.4pt);
  \draw[black,line width=0.7pt] (0,0) circle (3.2pt);


\end{tikzpicture}
\end{center}
\vskip 1em
\noindent
Each line corresponds to a copy of $\sl_2$ in $\sl_3^\gen$.
The centralizer in $\e_7$ of each such $\sl_2$ is isomorphic to $\so_{12}$.
We get three different $\sl_2 \oplus \so_{12}$ subalgebras of $\e_7$ this way.   
Their intersection is a copy of $\sl_6 \oplus \C^2$ that contains $\gsm$.  

Now to the main point: how to get three generations.
Each of the three $\sl_2 \oplus \so_{12}$ subalgebras is the direct sum of their common 
intersection $\sl_6 \oplus \C^2$ and a 32-dimensional subspace.    The 
subalgebra $\gsm \subset \e_7$ acts on each of these 32-dimensional subspaces 
via the $\e_7$ Lie bracket.  In each case this action is equivalent to the Standard Model representation of $\gsm$ on $\Lambda \C^5$.  If we call these subspaces $V_1, V_2$ and $V_3$, we have
\[
\begin{array}{ccl}
  \e_7 &\cong & (\sl_6 \oplus \C^2) \,\oplus\, V_1 \,\oplus\, V_2 \,\oplus\, V_3.
\end{array}
\]
The Lie algebra $\gsm$ is contained in $\sl_6 \oplus \C^2$.    We thus get an inclusion 
\[       \gsm \oplus\, V_1 \,\oplus\, V_2 \,\oplus\, V_3 \; \subset \; \e_7 \]
where the $\e_7$ Lie bracket gives three copies of the Standard Model representation of $\gsm$.

All this beauty comes at a price.   The three generations only account for $3 \cdot 32 = 96$ dimensions of $\e_7$.   The rest are accounted for by $\sl_6 \oplus \C^2$, which has dimension $35+ 2 = 37 = 133 - 96$.  We can split $\sl_6 \oplus \C^2$ into a $7$-dimensional Cartan subalgebra and $30$ root spaces.  This Lie algebra contains $\gsm$, which has a $4$-dimensional Cartan and  $8$ root spaces.   The remainder consists of $3$ Cartan dimensions and Nasmith's ``$22$ additional root spaces'' \cite[\S6]{Nasmith2020}.  Together they span a $25$-dimensional representation of $\gsm$ and thus 
$\SU(3) \times \SU(2)$, as shown in Table \ref{tab:22rootspaces}.

There is a unique copy of $\sl_5$ in $\e_7$ between $\gsm$ and the aforementioned
$\sl_6$.   Indeed there is a chain
\[           \gsm \subset \sl_5 \subset \sl_6 \subset \sl_6 \oplus \C^2 \subset \e_7  \]
with compatible root systems for all these Lie algebras.
The first row in Table \ref{tab:22rootspaces} corresponds to the 3 Cartan dimensions
in $\e_7$ not contained in $\gsm$.  The second corresponds to the 12 root spaces of $\e_7$
 lying in $\sl_5$ but not in $\gsm$.   In an $\SU(5)$ grand unified theory, these correspond
to the extra gauge bosons not present in the Standard Model, the so-called $X$ and $Y$
bosons, which serve to convert quarks into leptons and vice versa.
The last two rows together form the
$\mathbf{5} \oplus \bar{\mathbf{5}}$ of $\sl_5$,  corresponding to the $10$ root spaces lying
in $\sl_6$ but not in $\sl_5$.  In an $\SU(5)$ grand unified theory, the $\mathbf{5}$ is a multiplet
that contains the Standard Model Higgs boson, and this multiplet has antiparticles forming the
$\bar{\mathbf{5}}$.     Of course $\bar{\mathbf{2}} \cong \mathbf{2}$ as representations of 
$\SU(2)$; we write the bar simply to highlight some symmetries in the table.

\begin{table}[t]
\begin{tabular}{@{}llr@{}}
$\SU(3) \times \SU(2)$ representation & $\SU(5)$ interpretation & $\dim$ \\
\midrule
$(\mathbf{1},\mathbf{1}) \oplus (\mathbf{1},\mathbf{1}) \oplus  (\mathbf{1},\mathbf{1})$ & $\SU(5)$ singlets & $3$ \\
$(\mathbf{3},\mathbf{2}) \oplus (\bar{\mathbf{3}}, \bar{\mathbf{2}})$
  & $X$, $Y$ leptoquark gauge bosons of $\SU(5)$ & $12$ \\
$(\mathbf{3},\mathbf{1}) \oplus (\mathbf{1},\mathbf{2})$
  & Higgs $\mathbf{5}$ & $5$ \\
$(\bar{\mathbf{3}},\mathbf{1}) \oplus (\mathbf{1},\bar{\mathbf{2}})$
  & Higgs $\bar{\mathbf{5}}$  & $5$ \\
\midrule
& & $25$ \\
\end{tabular}
\vskip 1em
  \caption{Summands of $\e_7$ not in $\gsm$ or the three generations.}  
  \label{tab:22rootspaces}
\end{table}

\subsection*{Relation to earlier work.}

Our work arose as an attempt to understand the work of Nasmith 
\cite{Nasmith2020,Nasmith2023}.   But it also has other antecedents.   Kugo and Yanagida
\cite{KugoYanagida1984} studied generation unification using supergravity on 
the manifold $\E_7/(\SU(5) \times \SU(3) \times \mathrm{U}(1))$ and
identified three families of quarks and leptons with pseudo Nambu--Goldstone
fermions, arguing that the triplication of generations points to an underlying
$\E_7$ symmetry.
Sato and Yanagida \cite{SatoYanagida1998} refined this to an
$\E_7/(\SU(5) \times \mathrm{U}(1)^3)$ model, noting in particular that 
the model supplies three right-handed neutrinos.  

Our work also has strong connections to Lisi's study of $\e_7$ using Clifford
algebras, the nonassociative algebra $\mathbb{H} \otimes \mathbb{O}$, 
and triality \cite{Lisi2026}.    For example, our three decompositions of $\e_7$ as a 
direct sum of a subalgebra $\sl_2 \oplus \so_{12}$ and a 
64-dimensional representation of this subalgebra echo his work, which uses
one such decomposition of a real form of $\e_7$.   
However, he attempts to include spin degrees of freedom, which doubles
the size of the Standard Model representation from $\Lambda \C^5$ to
\[  (\Lambda^\text{even} \C^5 \otimes \textbf{2}) \;\oplus\; (\Lambda^{\text{odd}} \C^5 \otimes 
\overline{\textbf{2}}) , \] 
where $\textbf{2}$ is the defining representation of $\sl_2$
and $\overline{\textbf{2}}$ is its conjugate.  For this reason he cannot fit 
three linearly independent generations of fermions and their antiparticles into $\e_7$.
Like Nasmith, we simply \emph{ignore} how the particles transform under Lorentz 
transformations: tackling this question would require a larger algebraic structure and
new ideas.   Distler and Garibaldi \cite{DistlerGaribaldi2010} have put strong limitations
on attempts to include the complexified Lorentz Lie algebra $\sl_2$ and chiral spin-1/2 fermions 
in any form of $\e_8$.   On the other hand, Krasnov and Percacci \cite{KrasnovPercacci2018}
point out that $(\Lambda^\text{even} \C^5 \otimes \textbf{2}) \;\oplus\; (\Lambda^{\text{odd}} \C^5 \otimes 
\overline{\textbf{2}})$ arises naturally from restricting the chiral spinor representation of
$\so_{14}$ to $\so_{10} \oplus \so_4$. 

Following Nasmith, our work relies heavily on the chain of Lie algebras shown in Table \ref{tab:e_n_series}, which arises from $\e_7$ by a systematic process of `removing
roots':
\[    \sl_3 \oplus \sl_2 \subset \sl_5 \subset \so_{10} \subset \e_6 \subset \e_7. \]
This chain also plays a key role in our work with Bokor and Boyle \cite{BaezBokorBoyle2026}.
There we exploited the fact that each step down the chain gives a hermitian Jordan triple system.
(This fact does not hold for $\e_7 \subset \e_8$.)  We showed that the chain going down from 
$\e_6$ naturally accounts for the representation of $\gsm$ on \emph{one} generation of 
Standard Model fermions and their antiparticles.   
We raised the possibility that the chain going down from $\e_7$ could account for 
\emph{three}, but did not settle the question.  The work here may shed some light on that.

\subsection*{Plan.}
Section~\ref{sec:setup} defines a `good' embedding of $\gsm$ in 
$\e_7$ and shows how to construct one.   From then on
we choose a particular good $\gsm \subset \e_7$ and work with that.
Section~\ref{sec:gen} shows that the centralizer of this $\gsm$ in $\e_7$ 
contains only one subalgebra isomorphic to $\sl_3$, which we call $\sl_3^\gen$.   
Section~\ref{sec:three} studies three copies of $\sl_2$ in $\sl_3^\gen$, which we
call $\sl_2(\beta_k)$ for $k = 1,2,3$.  Section~\ref{sec:tri} gives an analysis of the 
roots of $\e_7$ crucial for all the proofs that follow.  Section~\ref{sec:so12} shows that the
centralizer of each subalgebra $\sl_2(\beta_k)$ is a subalgebra isomorphic to $\so_{12}$, 
which we call $\so_{12}(\beta_k)$.   It also studies three subalgebras 
$\m_k = \sl_2(\beta_k) \oplus \so_{12}(\beta_k) \subset \e_7$.
 Section~\ref{sec:sl6} shows that the centralizer of $\sl_3^\gen$
in $\e_7$ is a copy of $\sl_6$ lying in the intersection of all three of these subalgebras 
$\m_k$.   We call this centralizer $\sl_6^\sm$.   Section~\ref{sec:sl5} show that there is
a unique copy of $\sl_5$ containing $\gsm$ and contained in $\sl_6^\sm$.   We call this 
$\sl_5^\sm$.   

Based on all this work, Section~\ref{sec:generations} constructs three 30-dimensional
subspaces of $\e_7$ that transform under $\gsm$ as one generation of fermions and antifermions \emph{not including} right-handed neutrinos and their antiparticles.  Section~\ref{sec:lambda} 
shows how to construct larger 32-dimensional subspaces $V_1, V_2, V_3$ that transform under 
$\gsm$ as one generation of fermions including right-handed neutrinos and their antiparticles.  

\section{The Standard Model Lie algebra in $\e_7$}\label{sec:setup}

All our constructions rely on choosing a Lie subalgebra of $\e_7$ that is 
isomorphic to $\gsm$.    We need to choose one with good properties.   If one 
such subalgebra counts as `good' and it is mapped to another by any automorphism of $\e_7$,
the other also counts as `good'.   However, there are probably many
subalgebras of $\e_7$ that are isomorphic to $\gsm$ but 
not mapped to each other by automorphisms of $\e_7$.     For example, it 
is known that there is a countable infinity of 
subalgebras of $\so_{10}$ isomorphic to $\gsm$, not related to each other by
automorphisms of $\so_{10}$ \cite{Baez2025}.   Thus, we should be careful about our choice.

Here we describe two methods of choosing a good $\gsm$ subalgebra of
$\e_7$: one that is quick to describe, and another that
is more useful.   Both require the concept of a `regular' 
Lie subalgebra \cite[Chap.\ II]{Dynkin1957}.   
A subalgebra $\g'$ of a semisimple Lie algebra $\g$ with chosen Cartan $\h$ is 
\define{regular} if $[\h,\g'] \subseteq \g'$.  Regular subalgebras are a powerful tool
because they are easy to work with using root systems. 

For the quick method, start with the subgroup $\GSM \subset \SU(5)$ given 
in Equation \eqref{eq:GSM}.  Taking Lie
algebras and complexifying, we obtain the Lie subalgebra $\gsm \subset \sl_5$, and this is 
regular with respect to the standard Cartan of $\sl_5$.    Then choose any 
regular subalgebra of $\e_7$ isomorphic to $\sl_5$.   All such regular subalgebras are 
related by automorphisms of $\e_7$: this follows from the fact that all $\A_4$ root 
subsystems of $\E_7$ lie in a single orbit of the $\E_7$ Weyl group 
\cite[Table 10.2]{Oshima2006}.    Composing these two embeddings
\[            \gsm \subset \sl_5 \subset \e_7, \]
$\gsm$ becomes a regular subalgebra of $\e_7$.  Any subalgebra of $\e_7$ obtained
by this procedure counts as `good' for our purposes.

In the second method, we start with $\e_7$ and construct a sequence of smaller and
smaller regular subalgebras by successively removing dots from its Dynkin diagram until 
we are led to $\gsm$.    This is an instance of a general procedure which we call `removing
a root'.   

Suppose $\g$ is a semisimple
Lie algebra with Cartan $\h$.   Let $\Phi \subset \h^\ast$ be its root system.   For each $r \in \Phi$,
define the root space $\g_r$ by
\[     \g_r = \{ x \in \g \; : \; [a,x] = r(a) x  \text{ for all } a \in \h \} .\]
We then have
\[            \g = \h \; \oplus \; \bigoplus_{r \in \Phi} \g_r  .\] 
Choose a set of simple roots 
$\alpha_1, \dots, \alpha_n \in \Phi$.  We now trim down $\g$ to a smaller semisimple 
Lie algebra by removing the last of these simple roots, using the following procedure.

Let $\Phi' \subset \Phi$ be the intersection of $\Phi$ and the set of integer linear 
combinations of $\alpha_1, \dots, \alpha_{n-1}$.  Let $\h' \subset \h$ be
the subspace on which the root $\alpha_n \maps \h \to \C$ vanishes.  Then 
\[             \g' = \h' \; \oplus \; \bigoplus_{r \in \Phi'} \g_r \]
is a semisimple Lie algebra in its own right.  It is a regular subalgebra of $\g$, with rank 
one less than that of $\g$.   We take its Cartan to be $\h'$, and its root system is then $\Phi'$.   
This allows us to repeat the procedure as long as $\Phi'$ is nonempty.

It is also important to note that $\g'$ is contained in the larger Lie subalgebra
\[       \mathfrak{l} = \h \; \oplus \; \bigoplus_{r \in \Phi'} \g_r .\]
This is also a regular subalgebra of $\g$, but it is not semisimple, merely \define{reductive}:
that is isomorphic to the direct sum of a semisimple Lie algebra and an abelian one.  It is
isomorphic to $\g' \oplus \C$, since $\h$ is one dimension bigger than $\h'$.   This 
subalgebra $\mathfrak{l}$ is called a \define{maximal Levi subalgebra} of $\g$.   We 
shall see that $\gsm \cong \sl_3 \oplus \sl_2 \oplus \C$ arises as a maximal Levi subalgebra 
of $\sl_5$.

We start the procedure with $\e_7$.  Fix a Cartan subalgebra $\h \subset \e_7$ and
let $\Phi \subset \h^*$ be the resulting root system, so 
\[
  \e_7 \;=\; \h \,\oplus\, \bigoplus_{r \in \Phi} (\e_7)_r ,
\]
where each root space $(\e_7)_r$ has dimension one.  
Choose a set of simple roots $\alpha_1, \dots, \alpha_7 \in \Phi$ such that any pair lies at an angle of 120$^\circ$ when connected by an edge in this diagram, and 90$^\circ$ otherwise:

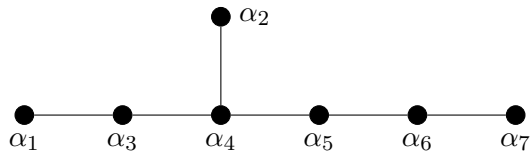
\begin{center}
\begin{tikzpicture}[
    node/.style={circle,draw,fill=white,inner sep=0pt,minimum size=7pt},
    every node/.style={},
    darknode/.style={circle,draw,fill=black,inner sep=0pt,minimum size=7pt},
    every node/.style={},
    x=1.3cm,y=1.3cm]

  \node[darknode] (a1) at (0,0) {};
  \node[darknode] (a3) at (1,0) {};
  \node[darknode] (a4) at (2,0) {};
  \node[darknode] (a5) at (3,0) {};
  \node[darknode] (a6) at (4,0) {};
  \node[darknode] (a7) at (5,0) {};
  \node[darknode] (a2) at (2,1) {};

  \draw (a1) -- (a3) -- (a4) -- (a5) -- (a6) -- (a7);
  \draw (a4) -- (a2);

  \node[below=3pt] at (a1) {$\alpha_1$};
  \node[below=3pt] at (a3) {$\alpha_3$};
  \node[below=3pt] at (a4) {$\alpha_4$};
  \node[below=3pt] at (a5) {$\alpha_5$};
  \node[below=3pt] at (a6) {$\alpha_6$};
  \node[below=3pt] at (a7) {$\alpha_7$};
  \node[right=3pt] at (a2) {$\alpha_2$};

\end{tikzpicture}
\end{center}

We then carry out the procedure just described, repeatedly, removing roots one at a time starting from $\alpha_7$ and working our way down to $\alpha_4$.  We obtain a series of regular semisimple Lie subalgebras of $\e_7$.  These are sometimes called the \define{$\e_n$ series}, or more precisely the portion of the $\e_n$ series going from $n = 7$ down to $n = 3$.   We show them in Table \ref{tab:e_n_series}, along with the dimension of each Lie algebra, its rank, and its number of roots.

At the last stage, the maximal Levi subalgebra is $\e_3 \oplus \C \cong \gsm$.   By construction this is a regular subalgebra of $\e_7$.   Any subalgebra of $\e_7$ that is carried to this one under some automorphism is defined to be \define{good}, and all the constructions of this paper apply to it.

In what follows we often work with a fixed Cartan $\h \subset \e_7$, which gives a Cartan for all the smaller Lie algebras in Table \ref{tab:e_n_series}, but we state as many results as we can in a Cartan-independent way, depending only on a chosen good $\gsm \subset \e_7$.   We 
identify $\h$ with its dual using the Killing form, normalized so that $\|r\|^2 = 2$ for every root: this is possible since the $\E_7$ root system is simply 
laced.    We let $V$ be the real span of $\Phi$, which is a 7-dimensional real inner product space contained in $\h$.   We have
\begin{equation}\label{eq:integrality}
  \ip{r}{s} \in \{-2,-1,0,1,2\} \qquad \text{for all } r,s \in \Phi 
\end{equation}
with the value $2$ only when $s = r$, and $-2$ only when $s = -r$.
   
\begin{table}[h]
  \label{tab:e_n_series}
  \centering
  \begin{tabular}{@{}llllp{0.38\textwidth}@{}}
    \toprule
      Dynkin diagram & semisimple Lie algebra  $\g$ & $\dim \g = \dim \h + |\Phi| $ \\
    \midrule
\raisebox{-0.5em}{\scalebox{0.5}{
\begin{tikzpicture}[
    node/.style={circle,draw,fill=white,inner sep=0pt,minimum size=7pt},
    every node/.style={},
    darknode/.style={circle,draw,fill=black,inner sep=0pt,minimum size=7pt},
    every node/.style={},
    x=1.3cm,y=1.3cm]

  \node[darknode] (a1) at (0,0) {};
  \node[darknode] (a3) at (1,0) {};
  \node[darknode] (a4) at (2,0) {};
  \node[darknode] (a5) at (3,0) {};
  \node[darknode] (a6) at (4,0) {};
  \node[darknode] (a7) at (5,0) {};
  \node[darknode] (a2) at (2,1) {};

  \draw (a1) -- (a3) -- (a4) -- (a5) -- (a6) -- (a7);
  \draw (a4) -- (a2);

  \node[below=3pt] at (a1) {$\alpha_1$};
  \node[below=3pt] at (a3) {$\alpha_3$};
  \node[below=3pt] at (a4) {$\alpha_4$};
  \node[below=3pt] at (a5) {$\alpha_5$};
  \node[below=3pt] at (a6) {$\alpha_6$};
  \node[below=3pt] at (a7) {$\alpha_7$};
  \node[right=3pt] at (a2) {$\alpha_2$};

\end{tikzpicture}
}}
& $\e_7$    
& 133 = 7 + 126
       \\  [6pt]
\raisebox{-0.5em}{\scalebox{0.5}{
\begin{tikzpicture}[
    node/.style={circle,draw,fill=white,inner sep=0pt,minimum size=7pt},
    every node/.style={},
    darknode/.style={circle,draw,fill=black,inner sep=0pt,minimum size=7pt},
    every node/.style={},
    x=1.3cm,y=1.3cm]

  \node[darknode] (a1) at (0,0) {};
  \node[darknode] (a3) at (1,0) {};
  \node[darknode] (a4) at (2,0) {};
  \node[darknode] (a5) at (3,0) {};
  \node[darknode] (a6) at (4,0) {};
  \node[node] (a7) at (5,0) {};
  \node[darknode] (a2) at (2,1) {};

  \draw (a1) -- (a3) -- (a4) -- (a5) -- (a6) -- (a7);
  \draw (a4) -- (a2);

  \node[below=3pt] at (a1) {$\alpha_1$};
  \node[below=3pt] at (a3) {$\alpha_3$};
  \node[below=3pt] at (a4) {$\alpha_4$};
  \node[below=3pt] at (a5) {$\alpha_5$};
  \node[below=3pt] at (a6) {$\alpha_6$};
  \node[below=3pt] at (a7) {$\alpha_7$};
  \node[right=3pt] at (a2) {$\alpha_2$};

\end{tikzpicture}
}} 
& $\e_6$ 
&  78 = 6 + 72
  \\ [6pt]
\raisebox{-0.5em}{\scalebox{0.5}{
\begin{tikzpicture}[
    node/.style={circle,draw,fill=white,inner sep=0pt,minimum size=7pt},
    every node/.style={},
    darknode/.style={circle,draw,fill=black,inner sep=0pt,minimum size=7pt},
    every node/.style={},
    x=1.3cm,y=1.3cm]

  \node[darknode] (a1) at (0,0) {};
  \node[darknode] (a3) at (1,0) {};
  \node[darknode] (a4) at (2,0) {};
  \node[darknode] (a5) at (3,0) {};
  \node[node] (a6) at (4,0) {};
  \node[node] (a7) at (5,0) {};
  \node[darknode] (a2) at (2,1) {};

  \draw (a1) -- (a3) -- (a4) -- (a5) -- (a6) -- (a7);
  \draw (a4) -- (a2);

  \node[below=3pt] at (a1) {$\alpha_1$};
  \node[below=3pt] at (a3) {$\alpha_3$};
  \node[below=3pt] at (a4) {$\alpha_4$};
  \node[below=3pt] at (a5) {$\alpha_5$};
  \node[below=3pt] at (a6) {$\alpha_6$};
  \node[below=3pt] at (a7) {$\alpha_7$};
  \node[right=3pt] at (a2) {$\alpha_2$};

\end{tikzpicture}
}}
&    $\e_5 = \so_{10}$ 
&   45 = 5 + 40 
\\ [6pt]
\raisebox{-0.5em}{\scalebox{0.5}{
\begin{tikzpicture}[
    node/.style={circle,draw,fill=white,inner sep=0pt,minimum size=7pt},
    every node/.style={},
    darknode/.style={circle,draw,fill=black,inner sep=0pt,minimum size=7pt},
    every node/.style={},
    x=1.3cm,y=1.3cm]

  \node[darknode] (a1) at (0,0) {};
  \node[darknode] (a3) at (1,0) {};
  \node[darknode] (a4) at (2,0) {};
  \node[node] (a5) at (3,0) {};
  \node[node] (a6) at (4,0) {};
  \node[node] (a7) at (5,0) {};
  \node[darknode] (a2) at (2,1) {};

  \draw (a1) -- (a3) -- (a4) -- (a5) -- (a6) -- (a7);
  \draw (a4) -- (a2);

  \node[below=3pt] at (a1) {$\alpha_1$};
  \node[below=3pt] at (a3) {$\alpha_3$};
  \node[below=3pt] at (a4) {$\alpha_4$};
  \node[below=3pt] at (a5) {$\alpha_5$};
  \node[below=3pt] at (a6) {$\alpha_6$};
  \node[below=3pt] at (a7) {$\alpha_7$};
  \node[right=3pt] at (a2) {$\alpha_2$};

\end{tikzpicture}
}}
&  $\e_4 = \sl_5$ 
&  24 = 4 + 20
     \\ [6pt] 
\raisebox{-0.5em}{\scalebox{0.5}{
\begin{tikzpicture}[
    node/.style={circle,draw,fill=white,inner sep=0pt,minimum size=7pt},
    every node/.style={},
    darknode/.style={circle,draw,fill=black,inner sep=0pt,minimum size=7pt},
    every node/.style={},
    x=1.3cm,y=1.3cm]

  \node[darknode] (a1) at (0,0) {};
  \node[darknode] (a3) at (1,0) {};
  \node[node] (a4) at (2,0) {};
  \node[node] (a5) at (3,0) {};
  \node[node] (a6) at (4,0) {};
  \node[node] (a7) at (5,0) {};
  \node[darknode] (a2) at (2,1) {};

  \draw (a1) -- (a3) -- (a4) -- (a5) -- (a6) -- (a7);
  \draw (a4) -- (a2);

  \node[below=3pt] at (a1) {$\alpha_1$};
  \node[below=3pt] at (a3) {$\alpha_3$};
  \node[below=3pt] at (a4) {$\alpha_4$};
  \node[below=3pt] at (a5) {$\alpha_5$};
  \node[below=3pt] at (a6) {$\alpha_6$};
  \node[below=3pt] at (a7) {$\alpha_7$};
  \node[right=3pt] at (a2) {$\alpha_2$};

\end{tikzpicture}
}}
&  $\e_3 = \sl_3 \oplus \sl_2$ 
& 11 = 3 + 8
     \\
  \end{tabular}
\vskip 1em
  \caption{The $\e_n$ series from $n = 7$ down to $n = 3$.}
\end{table}

\section{The generation \texorpdfstring{$\sl_3$}{sl(3)}}\label{sec:gen}

\begin{proposition} \label{prop:sl_3}
The centralizer of $\gsm$ in $\e_7$ is isomorphic to $\sl_3 \oplus \C^2$,
and it contains exactly one subalgebra isomorphic to $\sl_3$.
\end{proposition}

\begin{proof}
This description of the centralizer is in \cite[\S6]{Nasmith2020}.  We prove it
here using the general theory of regular subalgebras \cite{Dynkin1957}, since
the ideas are also used in later centralizer computations.

Suppose $\g$ is a simply laced semisimple Lie algebra with Cartan subalgebra
$\h$ and root system $\Phi$.  Let $\g' \subseteq \g$ be a regular reductive
subalgebra, so 
\[  \g' = \h' \oplus \bigoplus_{r \in \Phi'} \g_r \] 
for some
subspace $\h' \subseteq \h$ and some subsystem $\Phi' \subseteq \Phi$ with
$\operatorname{span}(\Phi') \subseteq \h'$.  Note that $\h'$ may be strictly
larger than $\operatorname{span}(\Phi')$, so that $\g'$ may have a nontrivial 
center, as is the case for $\gsm$.

An element of $\g$ centralizes $\g'$ if and only if it commutes with $\h'$ and
with every root space $\g_r$ for $r \in \Phi'$.   One can 
show the centralizer $C_{\g}(\g')$ is spanned by
\begin{itemize}
\item the subspace of $\h$ orthogonal to every element of $\Phi'$, and
\item the root spaces $\g_r$ for those roots $r \in \Phi$ orthogonal to every element of $\Phi'$.
\end{itemize}
(If $\g$ were not simply laced, we would need to replace the second item by a stronger condition.)   We can summarize this result as
\begin{equation}
\label{eq:centralizer}
C_{\g}(\g') = \Phi'^\perp
   \;\oplus\; \bigoplus_{r \in \Phi \cap \Phi'^\perp} \g_r ,
\end{equation}
where $\Phi'^\perp$ denotes the subspace of $\h$ orthogonal to 
every element of $\Phi'$.

In the case at hand, $\g' = \gsm$ is a regular subalgebra of $\g = \e_7$, with 
root system $\Phi' \subset \Phi$, and
\[   C_{\e_7}(\gsm) =\Phi'^\perp \oplus \bigoplus_{r \in \Phi  \cap \Phi'^\perp} (\e_7)_r  .\]
Here $\Phi \cap \Phi'^\perp$ is a root system of type $\A_2$.   Thus, the second summand 
above and the 2-dimensional subspace of the first spanned by these $\A_2$ roots 
together form a copy of $\sl_3$.   Since $\h$ has dimension 7 while the span of 
$\Phi'$ has dimension 3, $\Phi'^\perp$ has dimension 4; removing the two 
dimensions spanned by the $\A_2$ roots leaves a 2-dimensional abelian Lie algebra 
that commutes with this $\sl_3$.  Thus $C_{\e_7}(\gsm) \cong \sl_3 \oplus \C^2$ as desired.  

To show that the $\sl_3$ subalgebra is unique, let
$\mathfrak{a} \subseteq \sl_3 \oplus \C^2$ be any subalgebra isomorphic to
$\sl_3$, and let $p$ denote projection onto the abelian summand $\C^2$. Then
$p(\mathfrak{a})$ is an abelian quotient of $\mathfrak{a}$; since $\sl_3$ is
simple, $p(\mathfrak{a}) = 0$. Hence $\mathfrak{a} \subseteq \sl_3$, and
equality follows by comparing dimensions.
\end{proof}

We call the unique copy of $\sl_3$ in the centralizer of $\gsm$ the \define{generation $\sl_3$},
and denote it as $\sl_3^\gen$.   The proof of the proposition leads to this picture:
\vskip 1em
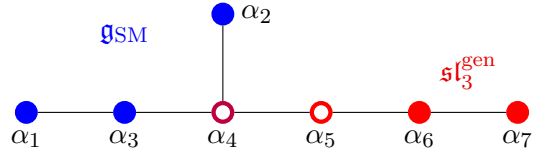
\begin{center}
\begin{tikzpicture}[
    node/.style={circle,draw,ultra thick,fill=white,inner sep=0pt,minimum size=7pt},
    every node/.style={},
    bluenode/.style={circle,draw,ultra thick,blue,fill=blue,inner sep=0pt,minimum size=7pt},
    every node/.style={},
    lightpurplenode/.style={circle,draw,ultra thick,purple,inner sep=0pt,minimum size=7pt},
    every node/.style={},
    lightrednode/.style={circle,draw,ultra thick,red,inner sep=0pt,minimum size=7pt},
    every node/.style={},
    rednode/.style={circle,draw,ultra thick,red,fill=red,inner sep=0pt,minimum size=7pt},
    every node/.style={},
    x=1.3cm,y=1.3cm]

  \node[bluenode] (a1) at (0,0) {};
  \node[bluenode] (a3) at (1,0) {};
  \node[lightpurplenode] (a4) at (2,0) {};
  \node[lightrednode] (a5) at (3,0) {};
  \node[rednode] (a6) at (4,0) {};
  \node[rednode] (a7) at (5,0) {};
  \node[bluenode] (a2) at (2,1) {};

  \draw (a1) -- (a3) -- (a4) -- (a5) -- (a6) -- (a7);
  \draw (a4) -- (a2);

  \node[] at (1,0.8) {\color{blue} $\gsm$};
  \node[] at (4.5,0.4) {\color{brickred} $\sl_3^\gen$};
 
  \node[below=3pt] at (a1) {$\alpha_1$};
  \node[below=3pt] at (a3) {$\alpha_3$};
  \node[below=3pt] at (a4) {$\alpha_4$};
  \node[below=3pt] at (a5) {$\alpha_5$};
  \node[below=3pt] at (a6) {$\alpha_6$};
  \node[below=3pt] at (a7) {$\alpha_7$};
  \node[right=3pt] at (a2) {$\alpha_2$};

\end{tikzpicture}
\end{center}
\vskip 0.5em
Each root $\alpha_i$ corresponds to a copy of $\sl_2$ in $\e_7$ with basis elements
$e_i, f_i, h_i$ obeying the usual relations
\[         [e_i, f_i] = h_i, \quad [h_i, e_i] = 2e_i, \quad [h_i, f_i] = -2f_i  .\]
The blue dots correspond to the semisimple part of $\gsm$: 
that is, $e_i, f_i$ and $h_i$ for $i = 1,2,3$ generate the subalgebra 
$\sl_3 \oplus \sl_2 \subset \gsm$.
The red dots correspond to the semisimple part of
the centralizer of $\gsm$, which is the generation $\sl_3$: $e_i, f_i$, and $h_i$ for $i = 6,7$ 
generate $\sl_3^\gen$.
The hollow purple and red dots correspond to the center of $C_{\e_7}(\gsm)$: $h_4$ and $h_5$ span this center, which is  a copy of $\C^2$ .  The hollow purple dot also corresponds to the
center of $\gsm$, which is spanned by $h_4$.   The overlap is no coincidence: by definition,
the center of $\gsm$ is the intersection of $\gsm$ and its centralizer.  
In physics, the center of $\gsm$ gives rise to hypercharge symmetry.

\section{Three copies of \texorpdfstring{$\sl_2$}{sl(2)} in the generation \texorpdfstring{$\sl_3$}{sl(3)}}\label{sec:three}

Being a regular subalgebra of $\e_7$, $\sl_3^{\gen}$ has a root system
$A \subset \Phi$ of type $\A_2$, spanning a plane $P$ in the real vector space $V$
spanned by the roots of $\e_7$.   We call $P$ the \define{generation plane}.
 It is convenient to name the weights of the defining representation of
$\sl_3^\gen$: these are $w_1, w_2, w_3 \in P$ with
\begin{equation}\label{eq:weights}
  w_1 + w_2 + w_3 = 0, \qquad |w_i|^2 = \tfrac{2}{3}, \qquad
  \ip{w_i}{w_j} = -\tfrac{1}{3} \;\; (i \neq j),
\end{equation}
and then
\[
  A \;=\; \{\, w_i - w_j \;:\; i \neq j \,\}, \qquad |A| = 6 .
\]
The six elements of $A$ lie on three lines through the origin. For $k \in \{1,2,3\}$ set
\begin{equation}\label{eq:betak}
  \beta_k \;:=\; w_i - w_j, \qquad i, j, k \text{ a cyclic permutation of } 1,2,3,
\end{equation}
so that $A = \{\pm\beta_1, \pm\beta_2, \pm\beta_3\}$. From
\eqref{eq:weights},
\begin{equation}\label{eq:orthogonality}
  \ip{w_k}{\beta_k} = 0, \qquad
  \ip{w_i}{\beta_k} = \pm 1 \;\; (i \neq k).
\end{equation}
Thus each line of $A$ singles out one of the three weights $w_k$: the one
orthogonal to it. 

\vskip 1em
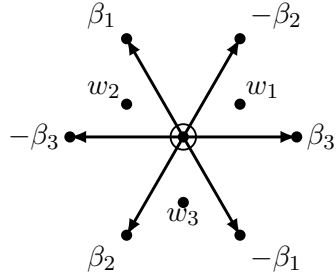
\begin{center}
\begin{tikzpicture}[scale=1.5,>=latex]

  \def\R{1}

  \foreach \a in {0,60,120,180,240,300}{
    \draw[->,line width=1.1pt,black] (0,0) -- (\a:\R);
    \fill (\a:\R) circle (1.4pt);
  }

  \fill (0,0) circle (1.4pt);
  \draw[black,line width=0.7pt] (0,0) circle (3.2pt);

  \def\rw{0.5774}   
  \fill[black] (30:\rw*\R)  circle (1.4pt);
  \fill[black] (150:\rw*\R) circle (1.4pt);
  \fill[black] (270:\rw*\R) circle (1.4pt);
  \node[anchor=south west,inner sep=2pt] at (30:\rw*\R)  {$w_1$};
  \node[anchor=south east,inner sep=2pt] at (150:\rw*\R) {$w_2$};
  \node[anchor=north,inner sep=2pt]      at (270:\rw*\R) {$w_3$};

  \node[anchor=west]       at (0:\R)   {$\beta_3$};
  \node[anchor=south east] at (120:\R) {$\beta_1$};
  \node[anchor=north east] at (240:\R) {$\beta_2$};
  \node[anchor=east]       at (180:\R) {$-\beta_3$};
  \node[anchor=north west] at (300:\R) {$-\beta_1$};
  \node[anchor=south west] at (60:\R)  {$-\beta_2$};

\end{tikzpicture}
\end{center}
\vskip 1em

Each line gives a copy of $\sl_2$ in the generation $\sl_3$:
\[
  \sl_2(\beta_k) \;=\; \C e_{\beta_k} \oplus \C h_{\beta_k}
                          \oplus \C f_{\beta_k}
                          \;\subset\; \sl_3^\gen ,
\]
We call these the three \define{generation $\sl_2$'s}.    As we shall see, the generation $\sl_3$ generates transformations that mix up fermions of all three generations.  The $k$th generation $\sl_2$ generates transformations that leave fermions in the $k$th generation fixed, while mixing up fermions in the other two generations.

\section{How the roots of \texorpdfstring{$\e_7$}{e7} see the plane
\texorpdfstring{$P$}{P}}\label{sec:tri}

Everything to come rests on the following trichotomy, which says that
a root of $\e_7$ can project down to the generation plane in only three qualitatively
different ways.    This is a reformulation of Nasmith's Lemma 2.1 \cite{Nasmith2020},
which in turn goes back to ideas in \cite{CGSS1976}.

Recall that $\Phi$ is the set of roots of $\e_7$, while $A \subset \Phi$ is the set of roots of $\gsm$.
Similarly $V$ is the real span of $\Phi$, while
$P \subset V$, the generation plane, is the real span of $A$.
Let $\pi \maps V \to P$ denote orthogonal projection.

\begin{lemma}\label{lem:trichotomy}
For every $r \in \Phi$,
\[
  \pi(r) \;\in\; \{0\} \;\sqcup\; \{\pm w_1, \pm w_2, \pm w_3\} \;\sqcup\; A ,
\]
and $\pi(r) \in A$ if and only if $r \in A$.
\end{lemma}

\begin{proof}
For $\delta \in A$ we have $\ip{\pi(r)}{\delta} = \ip{r}{\delta} \in \Z$ by
\eqref{eq:integrality}, since $\delta \in P$. As $A$ is simply laced with
$\delta^\vee = \delta$, this says precisely that $\pi(r)$ lies in the weight
lattice $\Lambda$ of $A$.

Now $\Lambda$ is the hexagonal lattice whose nonzero vectors have squared lengths
$\tfrac{2}{3}, \, 2, \, \tfrac{8}{3}, \dots$; the six vectors of squared length
$\tfrac{2}{3}$ are $\pm w_1, \pm w_2, \pm w_3$, and the six of squared length
$2$ are the roots $A$. Since $\pi$ is an orthogonal projection,
$|\pi(r)|^2 \le |r|^2 = 2$, so only $0$, $\pm w_i$ and $A$ can occur.

Finally, $|\pi(r)|^2 = 2 = |r|^2$ forces $r = \pi(r) \in P$. But
$\Phi \cap P$ is a rank-two subsystem of $\Phi$ containing $A$, and the only
rank-two simply laced system containing $\A_2$ is $\A_2$ itself; hence
$r \in A$.
\end{proof}

Now partition the roots outside $A$ according to their projection to the generation plane:
\begin{equation} \label{eq:Phi}
  \Phi_0 := \{\, r \in \Phi : \pi(r) = 0 \,\}, \qquad
  \Phi_k := \{\, r \in \Phi : \pi(r) = \pm w_k \,\} \quad (k = 1,2,3).
\end{equation}
By Lemma~\ref{lem:trichotomy} this is a genuine partition, so
\begin{equation}\label{eq:partition}
  \Phi \;=\; A \;\sqcup\; \Phi_0 \;\sqcup\; \Phi_1 \;\sqcup\; \Phi_2
             \;\sqcup\; \Phi_3 .
\end{equation}

\begin{lemma}\label{lem:counts}
$|\Phi_0| = 30$, and $\Phi_0$ is a root subsystem of type $\A_5$. Moreover
$|\Phi_1| = |\Phi_2| = |\Phi_3| = 30$.
\end{lemma}

\begin{proof}
The set $\Phi_0$ consists of the $\e_7$ roots orthogonal to the plane $P$.  This set spans 
$P^\perp \cap V$, which is $5$-dimensional.   To determine its cardinality
we can start by counting the roots of $\E_7$ orthogonal to any given root---it doesn't 
matter which one, since the Weyl group $W(\E_7)$ acts transitively on roots.  The  
answer is 60.   Since for each $k = 1,2,3$ we have
\[
  \{\, r \in \Phi : r \perp \beta_k \,\} \;=\; \Phi_0 \sqcup \Phi_k , \]
it follows that $|\Phi_0| + |\Phi_k| = 60$.  Summing over $k$ gives
\[
  3\,|\Phi_0| \;+\; \bigl(|\Phi_1| + |\Phi_2| + |\Phi_3|\bigr) \;=\; 180 .
\]
On the other hand the four set $\Phi_i$ partition $\Phi - A$
by Equation \eqref{eq:partition}, so
\[
  |\Phi_0| + |\Phi_1| + |\Phi_2| + |\Phi_3| \;=\; |\Phi| - |A| \;=\; 126 - 6
  \;=\; 120 .
\]
Subtracting the second identity from the first yields $2\,|\Phi_0| = 60$, hence $|\Phi_0| = 30$.
Since $|\Phi_0| + |\Phi_k| = 60$, we get
\[   |\Phi_1| = |\Phi_2| = |\Phi_3| \;=\; 30 . \]
Finally $\Phi_0$ is a root system in its own right, having rank $5$ and $30$ roots.  Among rank-five simply laced systems only $A_5$ has $30$ roots ($\D_5$ has $40$, and the largest reducible option $\D_4 \times \A_1$ has $26$), so $\Phi_0$ is of type $A_5$.
\end{proof}

We study the Lie subalgebra $\sl_6 \subset \e_6$ with root system $\Phi_0$ in Section 
\ref{sec:sl6}.

\section{Three copies of \texorpdfstring{$\so_{12}$}{so(12)} in \texorpdfstring{$\e_7$}{e7}}\label{sec:so12}

We now show that the three generation $\sl_2$'s are centralized by three different copies of $\so_{12}$.

\begin{lemma}\label{lem:d6}
For each $k$, $\Phi_0 \sqcup \Phi_k$ is a root subsystem of type $\D_6$.
\end{lemma}

\begin{proof}
Since $\beta_k \in P$, we have $r \perp \beta_k$ if and only if
$\pi(r) \perp \beta_k$, where $\pi \maps V \to P$ is the orthogonal projection
onto the generation plane.  Run through the possibilities of
Lemma~\ref{lem:trichotomy}: $0$ is orthogonal to $\beta_k$; $\pm w_k$ is
orthogonal to $\beta_k$ by Equation \eqref{eq:orthogonality}; $\pm w_i$ for $i \neq k$
is not, again by Equation \eqref{eq:orthogonality}; and no root of $A$ is orthogonal to
$\beta_k$, because in a system of type $\A_2$ any two roots have inner product
$\pm 2$ or $\pm 1$. Hence the orthogonal roots are exactly
$\Phi_0 \sqcup \Phi_k$, of cardinality $30 + 30 = 60$ by
Lemma~\ref{lem:counts}. The set is closed and has rank $6$; among rank-$6$
simply laced systems, $60$ roots identifies $\D_6$.  (In particular, $\A_6$ has $42$ and $\E_6$
has $72$.)
\end{proof}

Write $\so_{12}(\beta_k)$ for the corresponding regular subalgebra, with
Cartan subalgebra $\h \cap \beta_k^\perp$.

\begin{proposition}\label{prop:sl2so12}
For each $k \in \{1,2,3\}$ the subalgebras $\sl_2(\beta_k)$ and
$\so_{12}(\beta_k)$ are centralizers of each other in $\e_7$. Their direct sum
is a maximal-rank subalgebra
$\sl_2(\beta_k) \oplus \so_{12}(\beta_k)$ of $\e_7$, and
\begin{equation}\label{eq:sl2so12}
\e_7 \;=\; \sl_2(\beta_k) \,\oplus\, \so_{12}(\beta_k)
\,\oplus\, (\mathbf{2} \otimes \mathbf{32})
\end{equation}
where $\mathbf{2} \otimes \mathbf{32}$ is the tensor product of the defining
representation of $\sl_2(\beta_k)$ with a chiral spinor representation of
$\so_{12}(\beta_k)$.
\end{proposition}

\begin{proof}
By Lemma~\ref{lem:d6} the roots orthogonal to $\beta_k$ are
$\Phi_0 \sqcup \Phi_k$, of type $\D_6$; let $\so_{12}(\beta_k)$ be the
corresponding regular subalgebra, with Cartan subalgebra
$\h \cap \beta_k^\perp$. For $r \perp \beta_k$ one has
$|r \pm \beta_k|^2 = |r|^2 + |\beta_k|^2 = 4$, so $r \pm \beta_k \notin \Phi$
and thus $[\g_{\pm\beta_k}, \g_r] = 0$; the Cartan parts commute as well, so
$\sl_2(\beta_k)$ and $\so_{12}(\beta_k)$ commute with one another. Since
$\operatorname{rank}\sl_2 + \operatorname{rank}\so_{12} = 1 + 6 = 7 =
\operatorname{rank}\e_7$, the sum $\sl_2(\beta_k) \oplus \so_{12}(\beta_k)$
has maximal rank and each summand is the other's centralizer.

Write $M$ for the direct sum of the root spaces $\g_r$ for which
$r \not\perp \beta_k$ and $r \neq \pm\beta_k$, so that
\[ \e_7 = \sl_2(\beta_k) \oplus \so_{12}(\beta_k) \oplus M \]
as a representation of $\sl_2(\beta_k) \oplus \so_{12}(\beta_k)$.   The dimension
of $M$ is $133 - 3 - 66 = 64$.

First we identify $M$ as a representation of $\sl_2(\beta_k)$.  Because $\E_7$ is
simply laced, any two roots $r, \alpha$ satisfy
$\langle r, \alpha \rangle \in \{0, \pm 1, \pm 2\}$, with $\pm 2$ occurring only
for $r = \pm\alpha$.  Since the 64 roots $r$ contributing
to $M$ have $r \neq \pm\beta_k$ and $r \not\perp \beta_k$, these must have
$\langle r, \beta_k\rangle = \pm 1$. Thus $\ad(h_{\beta_k})$ has only the
eigenvalues $\pm 1$ on $M$, with eigenspaces $M_+$ and $M_-$, and the raising
and lowering operators $e_{\pm\beta_k}$ carry $M_\mp$ isomorphically onto
$M_\pm$. It follows that every irreducible $\sl_2(\beta_k)$-summand of $M$ is a copy of
the $2$-dimensional defining representation, which we call $\mathbf{2}$, and
\[
M \;\cong\; \mathbf{2} \otimes M_+ ,
\]
where $M_+$ is spanned by the $\e_7$ roots $r$ with $\langle r, \beta_k\rangle =1$.

It remains to identify $M_+$ as a representation of $\so_{12}(\beta_k)$.   It has dimension
$64/2 = 32$, and its $\so_{12}(\beta_k)$ weights 
are the orthogonal projections onto $\h \cap \beta_k^\perp$ of the $\e_7$ roots $r$ 
with $\langle r,\beta_k\rangle = 1$.  These $32$ weights are 
distinct, and each has the same length.  The only representations of $\so_{12}$ 
with these properties are the 
two chiral spinor (or `half-spin') representations, commonly called 
$\mathbf{32}$ and $\overline{\mathbf{32}}$.   These two  
are interchanged by the outer automorphism of $\so_{12}$ that swaps the two spinor nodes 
of the $\D_6$ diagram, so which one is $M_+$ depends on a convention about we identify $\so_{12}(\beta_k)$ with $\so_{12}$, and we choose a convention such that $M_+ \cong \textbf{32}$.
\end{proof}

From the proof we see that
\begin{equation} \label{eq:m_k_preliminary}
  \sl_2(\beta_k) \oplus \so_{12}(\beta_k) = \h \oplus \bigoplus_{r \in \{ \pm \beta_k\} \sqcup \Phi_0 \sqcup \Phi_k} (\e_7)_r 
\end{equation}
Thus, the three decompositions of $\e_7$ in Proposition \ref{prop:sl2so12} are truly 
different.  Proposition \ref{prop:intersection} says how they overlap.

\section{The standard \texorpdfstring{$\sl_6$}{sl(6)}} \label{sec:sl6}

There is a unique copy of $\sl_6$ lying in all three of the $\sl_2 \oplus 
\so_{12}$ subalgebras of $\e_7$ constructed in Proposition \ref{prop:sl2so12}.
This is the largest subalgebra of $\e_7$ that commutes with $\gsm$.

\begin{proposition} \label{prop:sl_6}
The centralizer of $\sl_3^\gen$ in $\e_7$ is isomorphic to $\sl_6$.   Its root system is $\Phi_0$.    
In fact $\sl_3^\gen$ and this copy of $\sl_6$ are centralizers of each other, and their
direct sum $\sl_3^\gen \oplus \sl_6$ is a maximal subalgebra of $\e_7$.  
\end{proposition}

\begin{proof}
Since $\sl_3^\gen$ is a regular simple subalgebra of $\e_7$, we can use same method 
as in Proposition \ref{prop:sl_3} to determine its centralizer: 
\[   C_{\e_7}(\sl_3^\gen) = P^\perp \oplus \bigoplus_{r \in \Phi_0} (\e_7)_r  \]
where $P^\perp$ is the orthogonal complement of the generation plane in $\h$, 
and $\Phi_0$ consists of all roots in $P^\perp$.
The root system $\Phi_0$ is of type $\A_5$ by Lemma \ref{lem:counts}.   Since $P$ has
dimension $2$, $P^\perp$ has dimension $7 - 2 = 5$, matching the rank
of $\A_5$, so there is no additional abelian summand and the centralizer is
isomorphic to $\sl_6$.

A similar argument shows that the centralizer of this $\sl_6$ is $\sl_3^\gen$.  Dynkin showed that 
for this $\sl_6 \subset \e_7$, the direct sum $\sl_3^\gen \oplus \sl_6$ is a maximal subalgebra
of $\e_7$: see \cite[Table 12]{Dynkin1957} and refer to Table 1 for his system of numbering roots
of $\e_7$.    In fact he described this maximal subalgebra using the procedure of removing a root,
described in Section \ref{sec:setup}, applied to the root we call $\alpha_5$.
\end{proof}

We call this copy of $\sl_6$ in $\e_7$ 
the \define{standard $\sl_6$}, and denote it as $\sl_6^\sm$.   

\begin{proposition}\label{prop:intersection}
Let $\m_k = \sl_2(\beta_k) \oplus \so_{12}(\beta_k)$.  Then
\[
  \m_1 \cap \m_2 \cap \m_3
  = \sl_6^\sm \, \oplus \, (\C \otimes P)
\]
for the 2-dimensional abelian Lie subalgebra $\C \otimes P \subset \e_6$, which is the 
complexification of the generation plane $P \subset \h$.  Moreover, the pairwise 
intersections of the $\m_k$ already equal the triple intersection.  
\end{proposition}

\begin{proof}
Recall from Equation \eqref{eq:m_k_preliminary} that
\begin{equation} \label{eq:m_k}
\m_k = \h \oplus \bigoplus_{r \in \{ \pm \beta_k\} \sqcup \Phi_0 \sqcup \Phi_k} (\e_7)_r .
\end{equation}
If $i \ne k$, any root $r \in \Phi_k$ lies outside $\{ \pm \beta_i\} \sqcup \Phi_0 \sqcup \Phi_i$, 
and so do $\pm \beta_k$.   Thus,
\[
   \m_i \cap \m_k = \h \oplus \; \bigoplus_{r \in \Phi_0}  (\e_7)_r.  
\]
This intersection is $\sl_6^\sm \oplus (\C \otimes P)$, 
since $\h$ is the direct sum of the subspace $\C \otimes P \subset \h$ 
and its orthogonal complement in $\h$, which is the Cartan subalgebra 
of $\sl_6^\sm$, while the roots of $\sl_6^\sm$ are exactly those in $\Phi_0$,
by Proposition \ref{prop:sl_6}.

Since this is true for all $i \ne k$, the triple intersection is
\begin{equation} \label{eq:intersection}
   \m_1 \cap \m_2 \cap \m_3 = \h \oplus \; \bigoplus_{r \in \Phi_0}  (\e_7)_r = \sl_6^\sm \oplus (\C \otimes P).  
\end{equation}
proving the proposition.
\end{proof}

\section{The standard $\sl_5$}\label{sec:sl5}

Sitting between $\sl_6^\sm$ and $\gsm$ there is a unique copy of $\sl_5$.   In some
ways this is more familiar to particle physicists than the 
previous constructions.   The inclusion $\GSM \subset \SU(5)$ from 
Equation \eqref{eq:GSM},  which is the basis of the Georgi--Glashow grand unified
theory \cite{GeorgiGlashow}, gives
an inclusion of complex Lie algebras $\gsm \subset \sl_5$ that matches the one here.   
However, we obtain it in a different way.

\begin{proposition}\label{prop:sl_5}
There is a unique subalgebra of $\sl_6^\sm$ isomorphic to $\sl_5$ that contains $\gsm$.
Moreover this $\sl_5$ subalgebra is regular in $\e_7$.  
\end{proposition}

\begin{proof}
We first construct a regular $\sl_5$ between $\gsm$ and $\e_7$, then show it lies
in $\sl_6^\sm$, and finally show it is the only $\sl_5$ between $\gsm$ and
$\sl_6^\sm$---regular or not.

\medskip
\emph{A regular $\sl_5$ containing $\gsm$.}
Since $\gsm$ is a regular subalgebra of $\e_7$, its Cartan subalgebra is
$\h_\sm = \h \cap \gsm$; let $U = \h_\sm \cap V$ be its real form, so that
$\dim_\R U = 4$ as $\gsm$ has rank $4$.

Let $\mathfrak{a}$ be any regular subalgebra of type $\A_4$ containing $\gsm$.
Being regular, $\mathfrak{a}$ has Cartan subalgebra $\mathfrak{a} \cap \h$.
This intersection contains $\h_\sm$, and both have dimension
$4$, so $\mathfrak{a} \cap \h = \h_\sm$.  Hence the roots of $\mathfrak{a}$ form a
subsystem of type $\A_4$ that lies in $\Phi \cap U$ and spans $U$. A direct
count in the coordinates of Section~\ref{sec:setup} shows $\Phi \cap U$ consists
of exactly $20$ roots and is itself a subsystem of type $\A_4$. Since a rank-$4$
subsystem of an $\A_4$ system is the whole system, the roots of $\mathfrak{a}$
are exactly $\Phi \cap U$. Thus there is at most one such $\mathfrak{a}$, namely
\[
  \sl_5^\sm \;:=\; \h_\sm \oplus \bigoplus_{r \in \Phi \cap U} (\e_7)_r,
\]
and one can check directly that this is a subalgebra of type $\A_4$; it is the
unique regular $\sl_5$ in $\e_7$ containing $\gsm$.

\medskip
\emph{It lies in $\sl_6^\sm$.}
We must show $\Phi \cap U \subseteq \Phi_0$. The generation plane $P$ is spanned
by the roots $A$, and $\h_\sm$ is orthogonal to $P$: the center of
$\gsm$ and the roots of $\sl_2 \oplus \sl_3 \subset \gsm$ all lie in
$C_{\e_7}(\sl_3^\gen)$, whose Cartan is $\h \cap P^\perp$ by
Proposition~\ref{prop:sl_6}. Hence $U \subseteq P^\perp$, so every root in
$\Phi \cap U$ is orthogonal to $P$, i.e.\ lies in $\Phi_0$. Therefore
$\sl_5^\sm \subseteq \sl_6^\sm$.

\medskip
\emph{Uniqueness in $\sl_6^\sm$.}
Now let $\s \subseteq \sl_6^\sm$ be \emph{any} subalgebra isomorphic to $\sl_5$
with $\gsm \subseteq \s$; we show $\s = \sl_5^\sm$, so that the regularity
hypothesis was superfluous. Identify $\sl_6^\sm$ with $\sl(\C^6)$ acting on its
defining representation $\C^6$. Since $\sl_5$ has no faithful representation of
dimension $6$ and its smallest nontrivial irreducible representations have
dimension $5$, the action of $\s$ on $\C^6$ decomposes as
\[
  \C^6 \;=\; W \oplus L, \qquad \dim W = 5, \quad \dim L = 1,
\]
with $W$ a faithful irreducible representation of $\s$ and $L$ trivial. Then
$\s \subseteq \sl(W) \subseteq \sl(\C^6)$, and comparing dimensions
($\dim \s = 24 = \dim \sl(W)$) gives $\s = \sl(W)$.  Thus, $\s$ is determined by the
hyperplane $W \subset \C^6$, and $\gsm \subseteq \s$ acts as zero on the line
$L$.

We claim the choice of the line $L$ is forced, and with it $W$ and $\s$. Under $\gsm$, 
the defining representation of $\sl_6$ decomposes as
\[
  \C^6\big|_{\gsm} \;=\; X \oplus Y \oplus T,
\]
where $X$ is a $3$-dimensional irreducible representation of $\sl_3 \subseteq
\gsm$ on which $\sl_2$ acts trivially, $Y$ is a $2$-dimensional irreducible
representation of $\sl_2 \subseteq \gsm$ on which $\sl_3$ acts trivially, and $T$
is a line on which $\sl_3$, $\sl_2$, and the center of $\gsm$ all act trivially.
Since $X$ and $Y$ are nontrivial irreducible representations of the simple summands
of $\gsm$, neither contains a nonzero $\gsm$-invariant vector; hence the space of
$\gsm$-invariants in $\C^6$ is exactly $T$. In particular the trivial
representation of $\gsm$ occurs in $\C^6$ with multiplicity one, so the
$\gsm$-fixed line $L$ must equal $T$. Then $W = \C^6/L$ is determined, and so is
$\s = \sl(W)$.

There is therefore at most one $\sl_5$ between $\gsm$ and $\sl_6^\sm$. As
$\sl_5^\sm$ is one, it is the only one, completing the proof.
\end{proof}

We call the unique $\sl_5$ subalgebra containing $\gsm$ and contained in $\sl_6^\sm$
the \define{standard $\sl_5$},  and denote it as $\sl_5^\sm$.   This subalgebra is generated
by the elements $e_i, f_i$ and $h_i$ for $i = 1,2,3,4,$ so we can visualize it as follows:
\vskip 0.5em
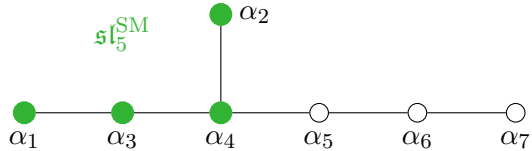
\begin{center}
\begin{tikzpicture}[
    node/.style={circle,draw,fill=white,inner sep=0pt,minimum size=7pt},
    every node/.style={},
    greennode/.style={circle,draw,darkgreen, ultra thick,fill=darkgreen,inner sep=0pt,minimum size=7pt},
    every node/.style={},
    lightpurplenode/.style={circle,draw,ultra thick,purple,inner sep=0pt,minimum size=7pt},
    every node/.style={},
    lightrednode/.style={circle,draw,ultra thick,red,inner sep=0pt,minimum size=7pt},
    every node/.style={},
    rednode/.style={circle,draw,ultra thick,brickred,fill=red,inner sep=0pt,minimum size=7pt},
    every node/.style={},
    x=1.3cm,y=1.3cm]

  \node[greennode] (a1) at (0,0) {};
  \node[greennode] (a3) at (1,0) {};
  \node[greennode] (a4) at (2,0) {};
  \node[node] (a5) at (3,0) {};
  \node[node] (a6) at (4,0) {};
  \node[node] (a7) at (5,0) {};
  \node[greennode] (a2) at (2,1) {};

  \draw (a1) -- (a3) -- (a4) -- (a5) -- (a6) -- (a7);
  \draw (a4) -- (a2);

  \node[] at (1,0.8) {\color{darkgreen} $\sl_5^\sm$};
 
  \node[below=3pt] at (a1) {$\alpha_1$};
  \node[below=3pt] at (a3) {$\alpha_3$};
  \node[below=3pt] at (a4) {$\alpha_4$};
  \node[below=3pt] at (a5) {$\alpha_5$};
  \node[below=3pt] at (a6) {$\alpha_6$};
  \node[below=3pt] at (a7) {$\alpha_7$};
  \node[right=3pt] at (a2) {$\alpha_2$};

\end{tikzpicture}
\end{center}

\section{Three generations omitting right-handed neutrinos}\label{sec:generations}

So far, starting from only a good choice of subalgebra $\gsm \subset \e_7$ as in Section \ref{sec:setup}, 
we have constructed a chain of subalgebras:
\[    \gsm \subset \sl_5^\sm\subset \sl_6^\sm \subset \e_7. \]
Each one is uniquely determined:
\begin{itemize} 
\item  The centralizer of $\gsm$ contains a unique copy of $\sl_3$, namely $\sl_3^\gen$.
(Proposition \ref{prop:sl_3}.)  
\item  The centralizer of $\sl_3^\gen$ is copy of $\sl_6$, namely $\sl_6^\sm$. 
(Proposition \ref{prop:sl_6}.)
\item   There is a unique copy of $\sl_5$ between $\gsm$ and $\sl_6^\sm$, 
namely $\sl_5^\sm$.    (Proposition \ref{prop:sl_5}.)
\end{itemize}
We are now ready to exploit this chain.    

For each $k = 1,2,3$ we identify the 30-dimensional space spanned by the root spaces $(\e_7)_r$ with $r \in \Phi_k$ as a representation first of $\sl_6^\sm$, then of $\sl_5^\sm$, and finally of $\gsm$.  In the end, we obtain the usual representation of $\gsm$ on one generation of Standard Model fermions together with their antiparticles, \emph{except for} the right-handed neutrino and its antiparticle.

\begin{lemma}\label{lem:phik}
For each $k = 1,2,3$ we have
\[   \bigoplus_{r \in \Phi_k} (\e_7)_r  \;\cong\; \Lambda^2\C^6 \,\oplus\, \Lambda^4\C^6  \]
as representations of $\sl_6^\sm$.
\end{lemma}

\begin{proof}
For short let us write
\[   N_k =  \bigoplus_{r \in \Phi_k} (\e_7)_r  .\]
Note that $N_k$ is an $\sl_6^\sm$-submodule of $\e_7$: for $s \in \Phi_0$ and 
$r \in \Phi_k$ we have $\pi(r+s) = \pi(r) = \pm w_k$, so whenever $r + s$ is a root it again lies in
$\Phi_k$, showing that $[(\e_7)_s, (\e_7)_r] \subseteq N_k$.

Split $\Phi_k$ into two subsets, defining
\[
  \Phi_k^{\pm} \;=\; \{\, r \in \Phi_k : \pi(r) = \pm w_k \,\},
  \qquad \Phi_k = \Phi_k^{+} \sqcup \Phi_k^{-},
\]
and set 
\[   N_k^\pm = \bigoplus_{r \in \Phi_k^\pm} (\e_7)_r .\]
Since the projection $\pi \maps V \to P$ is unaffected by adding a root of
$\Phi_0$, each $N_k^\pm$ is itself an
$\sl_6^\sm$-submodule, and $N_k = N_k^+ \oplus N_k^-$.  The reflection
$r \mapsto -r$ carries $\Phi_k^+$ bijectively onto $\Phi_k^-$, so
$|\Phi_k^+| = |\Phi_k^-| = \tfrac12|\Phi_k| = 15$ by Lemma~\ref{lem:counts},
and $N_k^-$ is the dual of $N_k^+$.

$N_k^+$ is a 15-dimensional representation of $\sl_6^\sm$ with 15 distinct weights 
all of the same  length: namely, the roots of $\Phi_k^+$ projected down to the Cartan of 
$\sl_6^\sm$.   The only representations of $\sl_6$ with these properties are its fundamental
representations on $\Lambda^2\C^6$ and its dual $\Lambda^4\C^6$.   We must thus have
$N_k \cong \Lambda^2\C^6 \,\oplus\, \Lambda^4\C^6$.
\end{proof}

Next we restrict along the inclusion $\sl_5^\sm \subset \sl_6^\sm$.  

\begin{theorem} \label{thm:gen}
For each $k = 1,2,3$ we have
\[
\bigoplus_{r \in \Phi_k} (\e_7)_r \cong \Lambda^1\C^5 \oplus \Lambda^2\C^5 \oplus \Lambda^3\C^5 \oplus \Lambda^4\C^5
\]
as representations of $\sl_5^\sm$.
\end{theorem}

\begin{proof}
By Lemma~\ref{lem:phik} the span is $\Lambda^2\C^6 \oplus \Lambda^4\C^6$ as an
$\sl_6^\sm$-module.  Restricting to $\sl_5^\sm \subset \sl_6^\sm$ and splitting
$\C^6$ as $\C^5 \oplus \C$, we have
\[
\begin{array}{ccl}
  \Lambda^2(\C^5 \oplus \C) &\cong & (\Lambda^0\C^5 \otimes \Lambda^2 \C) \oplus 
(\Lambda^1\C^5 \otimes \Lambda^1 \C) \oplus 
(\Lambda^2\C^5 \otimes \Lambda^0 \C) \\ [3pt]
&\cong& \Lambda^1 \C^5 \oplus \Lambda^2\C^5 
\end{array}
\]
as representations of $\sl_5^\sm$, and dually 
\[
  \Lambda^4\C^6 \;\cong\; \Lambda^4\C^5 \oplus \Lambda^3\C^5
\]
so the proposition follows.
\end{proof}

Restricting this representation to $\gsm \subset \sl_5^\sm$ we obtain the
representation of $\gsm$ on one generation of Standard Model fermions together
with their antiparticles \emph{omitting} the right-handed neutrino and its antiparticle,
which transform trivially under $\gsm$ and would naturally lie in $\Lambda^0 \C^5$
and $\Lambda^5 \C^5$.   For details of how $\Lambda \C^5$ captures one generation 
of  fermions and their antiparticles see \cite[Table 4]{BaezHuerta2010}, which we reproduce 
here in Table \ref{tab:su5code}.   In this table we give $\C^5$ a basis called $r,g,b,u,d$.  We write 
wedge products simply as products, write $c$ for $r, g$ or $b$, and write $\cbar$ for 
$gb, br$ or $rg$.

\vskip 1em
	\begin{table}[h]
\begin{center} 
	\begin{tabular}{cccccc} 
		\toprule
		$\Lambda^0 \C^5$   & $\Lambda^1 \C^5$    & $\Lambda^2 \C^5$            & $\Lambda^3 \C^5$             & $\Lambda^4 \C^5$              & $\Lambda^5 \C^5$ \\ 
		\midrule
		$\nubar_L = 1$ & $e^+_R = u$     & $e^+_L = ud$            & $e^-_R = rgb$            & $e^-_L = drgb$            & $\nu_R = udrgb$ \\  \\
		               & $\nubar_R = d$  & $u^c_L = uc$            & $\ubar^\cbar_R = d\cbar$ & $\nu_L = urgb$            & \\  \\
			       & $d^c_R = c$     & $d^c_L = dc$            & $\dbar^\cbar_R = u\cbar$ & $\dbar^\cbar_L = ud\cbar$ & \\   \\
			       &                 & $\ubar^\cbar_L = \cbar$ & $u^c_R = udc$            &                           & \\   \\
	\vspace{-10pt}
	\end{tabular}
		\caption{First-generation fermions as exterior algebra elements.
      } \label{tab:su5code}
\end{center}
\end{table}

The statement of Theorem \ref{thm:gen} depends crucially on a chosen Cartan $\h \subset \e_7$ 
and also chosen weights $w_k$ for $\sl_3^\gen$.   Next we state a result that depends only on our 
choice of $\gsm \subset \e_7$.   With only this structure, there is no way to separate out 
and name the three generations, since the generation symmetry is unbroken: $\sl_3^\gen$ commutes with $\gsm$.   We shall see that $e_7$ splits into four parts: $\sl_3^\gen$, its centralizer
$\sl_6^\sm$, and two 45-dimensional representations of $\sl_3^\gen \oplus \sl_6^\sm$.
Taken together, these two 45-dimensional representations account for three generations of fermions 
and their antiparticles---omitting only the right-handed neutrinos and left-handed antineutrinos, 
which transform trivially under $\gsm$.  

To state this result, let $\textbf{3}$ and $\textbf{3}^\ast$ be the two irreducible 3-dimensional representations of $\sl_3^\gen$, and let $\textbf{15}$ and $\textbf{15}^\ast$ be the two 
irreducible 15-dimensional representations of $\sl_6^\sm$.    In each case, saying which representation 
is which depends on arbitrary conventions, since the Dynkin diagram of each Lie subalgebra of
$\e_7$ has a $\Z_2$ symmetry that exchanges dual representations.  

\begin{theorem} \label{thm:unified_gen}
As a representation of its maximal subalgebra $\sl_3^\gen \oplus \sl_6$, the Lie algebra
$\e_7$ decomposes into irreducible representations as follows:
\[        \e_7 = \sl_3^\gen \, \oplus \, \sl_6^\sm \, \oplus \, (\mathbf{3} \otimes \mathbf{15}^\ast) \,
\oplus \, (\mathbf{3}^\ast \otimes \mathbf{15}). \]
\end{theorem}

\begin{proof}
As a purely representation-theoretic result this is well known \cite[Table 52]{Slansky1981}, but we
outline a self-contained argument.    Clearly $\sl_3^\gen$ and $\sl_6^\sm$
are irreducible summands of $\e_7$ as a representation of $\sl_3^\gen \oplus \sl_6$.
By the proof of Lemma \ref{lem:phik}, the rest breaks up into two 45-dimensional summands,
\[    \textstyle{\bigoplus_{k = 1}^3} N_k^+ \; \oplus \; \textstyle{\bigoplus_{k = 1}^3} N_k^-  . \]
The proof also shows that $\sl_6^\sm$ has equivalent 15-dimensional irreducible
representations on each space $N_k^+$, and dual representations on each space $N_k^-$.   
On the other hand, $\sl_3^\gen$ acts irreducibly via its defining representation on the 
copy of $\C^3$ whose basis is the three roots $w_k$, and via the dual of the defining 
representation on the copy of $\C^3$ whose basis is the three roots $-w_k$.   
\end{proof}

By the proof of Theorem~\ref{thm:gen}, restricting one of the 15-dimensional irreducible 
representations of $\sl_6^\sm$ to $\sl_5^\sm$ gives $\Lambda^1 \C^5 \oplus \Lambda^2 \C^5$,
while the other gives $\Lambda^3 \C^5 \oplus \Lambda^4 \C^5$.  
Each 45-dimensional summand therefore contains a mixture of left-handed particles 
(meaning elements of $\Lambda^p \C^5$ with $p$ even, as shown in Table \ref{tab:su5code}) 
and right-handed particles (elements of $\Lambda^p \C^5$ with $p$ odd).

\section{Three generations in \texorpdfstring{$\e_7$}{e_7}}\label{sec:lambda}
 
We can also include right-handed neutrinos and their antiparticles in this setup. 
They behave unlike the other particles.  However, we can include them in such a
way that we get three copies of the Standard Model representation in $\e_7$.   

Theorem \ref{thm:gen} implies that unlike the other fermions of the $k$th generation, 
the $k$th right-handed neutrino cannot fit into the root spaces associated to $\Phi_k$.  
However, Nasmith \cite{Nasmith2020} noted that it can be fit into the root space associated to 
$\beta_k$, one of the six roots of the generation $\sl_3$.   Its antiparticle then naturally goes 
into the root space associated to $-\beta_k$.   Three copies of the Standard Model 
representation of $\gsm$ on $\Lambda \C^5$ then fill up all of $\e_7$ except for a copy
of $\sl_6 \oplus \C^2$.  The copy of $\sl_6$ here is $\sl_6^\gen$, while the copy of $\C^2$
is $\C \otimes P$, the complexification of the generation plane.

To implement this, for $k = 1,2,3$ let 
\begin{equation}\label{eq:Vk}
  V_k =  \bigoplus_{r \in \{\pm \beta_k\} \sqcup \Phi_k} (\e_7)_r.
\end{equation}
This space is 32-dimensional, and by Equations \eqref{eq:m_k} and \eqref{eq:intersection} we have
\begin{equation}
  \m_k = (\m_1 \cap \m_2 \cap \m_3) \oplus V_k .
\end{equation}
In fact $V_k$ captures one generation of Standard Model fermions and their antiparticles.

\begin{theorem}\label{thm:main}
As a representation of $\sl^\sm_6$,
\[
  V_k \;\cong\; \Lambda^0\C^6 \oplus \Lambda^2\C^6 \oplus \Lambda^4\C^6
                \oplus \Lambda^6\C^6 \;=\; \Lambda^{\mathrm{even}}\C^6.
\]
Consequently, as a representation of $\sl_5^\sm$ we have 
\[  V_k \;\cong\; \Lambda\C^5 ,\]
and as a representation of $\gsm$, $V_k$ is the Standard Model representation.
\end{theorem}

\begin{proof}
The first statement follows from Lemma~\ref{lem:phik} together with 
\[         (\e_7)_\beta \cong \Lambda^0 \C^6, \quad    (\e_7)_{-\beta} \cong \Lambda^6 \C^6.\]
For the second, split $\C^6 = \C^5 \oplus \C$ as $\sl_5$-modules. Since
$\Lambda(V_1 \oplus V_2) \cong \Lambda V_1 \otimes \Lambda V_2$, taking the even part and using
$\Lambda\C = \Lambda^0\C \oplus \Lambda^1\C$ gives
\[
  \Lambda^{\mathrm{even}}(\C^5 \oplus \C)
  \;\cong\; \bigl(\Lambda^{\mathrm{even}}\C^5 \otimes \Lambda^0\C\bigr)
     \oplus \bigl(\Lambda^{\mathrm{odd}}\C^5 \otimes \Lambda^1\C\bigr)
  \;\cong\; \Lambda^{\mathrm{even}}\C^5 \oplus \Lambda^{\mathrm{odd}}\C^5
  \;=\; \Lambda\C^5 . \qedhere
\]
\end{proof}

\begin{theorem}\label{thm:decomp}
We have a direct sum decomposition
\[
  \e_7 \;=\; \sl_6^\sm \, \oplus \, (\C \otimes P) \,\oplus\, V_1 \,\oplus\, V_2 \,\oplus\, V_3 .
\]
where each summand is an $\sl_6^\sm$-subrepresentation of $\e_7$, with $\sl_6^\sm$
acting by the $\e_7$ Lie bracket.
\end{theorem}

\begin{proof}
By Proposition \ref{prop:intersection} and Equation \eqref{eq:intersection} we have
\begin{equation}
\sl_6^\sm \oplus (\C \otimes P) = \h \oplus \; \bigoplus_{r \in \Phi_0}  (\e_7)_r.  
\end{equation}
By this formula and Equation \eqref{eq:Vk} it follows that the spaces $V_k$ are mutually
orthogonal and all orthogonal to $\sl_6^\sm \oplus (\C \otimes P)$, since by Equation 
\eqref{eq:Phi} the sets $\Phi_k$ ($k = 0,1,2,3$) and the sets $\{\pm \beta_k\}$ ($k = 1,2,3$)
are all disjoint.  Since the union of all these sets is the set $\Phi$ of $\e_7$ roots, we have
\[   \sl_6^\sm \oplus (\C \otimes P) \,\oplus\, V_1 \,\oplus\, V_2 \,\oplus\, V_3 \, = \, 
\h \oplus \, \bigoplus_{r \in \Phi} (\e_7)_r = \e_7 .\]
 
To show that the subspaces $V_k$ are preserved by $\sl_6^\sm$, 
first let $s \in \Phi_0$ and $r \in \Phi_k$: then $\pi(r+s) = \pi(r) = \pm w_k$, 
so if $r+s$ is a root it again lies in $\Phi_k$.  This shows that
\[  [(\e_7)_s,  \bigoplus_{r \in \Phi_k} (\e_7)_r ] \subseteq  \bigoplus_{r \in \Phi_k} (\e_7)_r,\]
so this direct sum is preserved by $\sl_6^\sm$.
Next let $s \in \Phi_0$ and consider $\beta_k$: if $\beta_k + s$ were a
root, its projection onto the generation plane would be $\beta_k$, of squared length $2$, forcing
$\beta_k + s \in P$ by the argument in Lemma~\ref{lem:trichotomy} and hence
$s = 0$.  Thus 
\[   [(\e_7)_s, (\e_7)_{\pm\beta_k}] = 0. \]
This shows that $V_k$ is preserved by $\sl_6^\sm$.

Clearly $\sl_6^\sm$ is preserved by bracketing with $\sl_6^\sm$.  Finally, 
$\C \otimes P$ is annihilated by bracketing with $\sl_6^\sm$ because it is spanned by the roots 
$\beta_k$, which lie in $\sl_3^\gen$, which commutes with $\sl_6^\sm$ by 
Proposition \ref{prop:sl_6}.
\end{proof}

The statement of this final theorem depends not only on the choice of a good $\gsm \subset \e_7$ but also a Cartan subalgebra $\h$ and three roots $\beta_k$ of $\sl_3^\gen$.   Theorem \ref{thm:unified_gen}, by comparison, depends only on $\gsm \subset \e_7$.   In the framework discussed here, the right-handed neutrinos and their antiparticles are fundamentally different from other fermions.  Not only do they transform trivially under $\gsm$, their 6-dimensional space cannot be
singled out from the 8-dimensional space $\sl_3^\gen$ without choosing a Cartan of that Lie algebra.
Proposition \ref{prop:sl_3} gives a 10-dimensional space of particles that transform trivially under
$\gsm$, and a natural way to pick out the 8-dimensional subspace $\sl_3^\gen$, but decomposing
it further requires extra choices.

\subsection*{Acknowledgements} 

I thank Benjamin Nasmith for bringing his work to my attention, and for helpful conversations.
I thank Curtis Laketek for catching an error.  
This work was done with help from Claude Opus 4.8: a friend gave me a subscription to 
Claude Pro and I wanted to test it out, despite my many misgivings, including how
LLMs are contributing to global warming and income inequality.   It was very good at
answering questions about Nasmith's paper, doing computations, and proving theorems.
When I asked it to summarize our discussions, it quickly produced a 10-page draft of this 
paper.  The proofs were terse and hard to understand.  For several weeks I checked, 
reorganized, expanded and completely rewrote this material.   The writing here is all my 
own, and any errors are mine alone. The process raised tough questions about how we should do 
mathematics.

\end{document}